\documentclass[aps,prl,twocolumn,notitlepage,superscriptaddress,showpacs,nofootinbib]{revtex4-1}

\usepackage{amsmath,amssymb,amsthm,mathtools}
\usepackage{booktabs,array,enumitem,microtype}
\usepackage{xcolor}
\usepackage[colorlinks=true,allcolors=blue!55!black]{hyperref}
\usepackage[T1]{fontenc}
\usepackage{mathptmx}

\setlist{nosep,leftmargin=*}

\newcommand{\F}{\mathbb F}
\newcommand{\C}{\mathbb C}
\newcommand{\Z}{\mathbb Z}
\newcommand{\cH}{\mathcal H}
\newcommand{\cU}{\mathcal U}
\newcommand{\cB}{\mathcal B}
\newcommand{\cC}{\mathcal C}
\newcommand{\cG}{\mathcal G}
\newcommand{\cV}{\mathcal V}
\newcommand{\BSEP}{\operatorname{BSEP}}
\newcommand{\PPT}{\operatorname{PPT}}
\newcommand{\suppcomp}{\operatorname{supp}_{\rm comp}}
\newcommand{\Span}{\operatorname{span}}
\newcommand{\ran}{\operatorname{ran}}

\newcommand{\Tr}{\operatorname{Tr}}
\newcommand{\ket}[1]{\lvert #1\rangle}
\newcommand{\bra}[1]{\langle #1\rvert}
\newcommand{\braket}[2]{\langle #1\vert #2\rangle}
\newcommand{\proj}[1]{\lvert #1\rangle\!\langle #1\rvert}
\newcommand{\transpose}{\mathsf T}
\newcommand{\abs}[1]{\lvert #1\rvert}
\newcommand{\norm}[1]{\lVert #1\rVert}

\newtheorem{theorem}{Theorem}
\newtheorem{lemma}{Lemma}
\newtheorem{proposition}{Proposition}
\newtheorem{corollary}{Corollary}

\theoremstyle{definition}

\makeatletter

\begin{document}
	
	\title{Genuinely Unextendible Product Bases from Maximum Distance Separable Codes}
	\author{Mao-Sheng Li}
	\email{li.maosheng.math@gmail.com}
	\affiliation{School of Mathematics, South China University of Technology, Guangzhou 510641, China}
	
	\begin{abstract}
	The existence of genuinely unextendible product bases (GUPBs), incomplete orthogonal sets of fully product states whose orthogonal complements contain no product vector across any bipartition, has remained an open problem. Here we construct GUPBs for any number $N\geq3$ of parties using classical maximum distance separable (MDS) codes. The MDS property imposes a rigidity on the induced product tiling across every bipartition; combined with Fourier mode deletion and a stopper state, this rigidity enforces genuine unextendibility. Consequently, the orthogonal complement of each GUPB is a genuinely entangled subspace whose normalized projector is invariant under partial transposition across every bipartition, yielding an explicit family of multipartite bound entangled states. We further construct GME witnesses that detect these states even though no fully decomposable witness can do so. Moreover, the resulting indistinguishability  persists under arbitrary finite tensor powers and measurements separable across any bipartition. These results establish a direct connection between error-correcting codes and multipartite entanglement and provide an algebraic route to certifying genuinely multipartite bound entanglement.

	\end{abstract}
	
	\maketitle
	
	\noindent\emph{Introduction.---}
	Orthogonal product states can exhibit nonclassical behavior even though none
	of the states is entangled.  Unextendible product bases (UPBs) are a canonical
	example: an incomplete orthogonal family of fully product vectors is
	unextendible when no further fully product vector is orthogonal to the whole
	family \cite{Bennett1999,DiVincenzo2003,Pittenger2003}.  The orthogonal
	complement is then a completely entangled subspace and its normalized
	projector yields a positive partial transpose (PPT) bound entangled state
	\cite{Bennett1999,Horodecki1998,BejHalder2021}.  UPBs are also tied to local
	state discrimination, data hiding, and Bell inequalities with no quantum
	violation \cite{BennettNWE1999,DeRinaldis2004,Hayden2005,Cohen2008,
		Augusiak2011,Augusiak2012}.  Their construction and classification have in
	turn led to combinatorial, graph-theoretic, and tile-based descriptions
	\cite{AlonLovasz2001,Feng2006,Johnston2014,ChenJohnston2015,
		ChenDjokovic2018,Yang2015,ShiTiles2020,WangChen2020,You2023,ShiGraph2023}.
	
	Multipartite systems sharpen the question.  A vector may be entangled over
	all individual parties and nevertheless factor across some bipartition
	$X|Y$.  A subspace containing no such vector for any nontrivial bipartition
	is a genuinely entangled subspace (GES)
	\cite{Parthasarathy2004,Bhat2006,Demianowicz2018,Agrawal2019,
		Demianowicz2019,Demianowicz2020,DemianowiczRajchel2021,Antipin2021,
		DemianowiczUniversal2022,JohnstonHierarchy2022}.  Accordingly, a
	\emph{genuinely unextendible product basis} (GUPB) is an orthogonal family of
	fully product vectors whose orthogonal complement is a GES, or equivalently,
	a family that admits no product extension across any bipartition.  This is
	strictly stronger than local irreducibility, strong nonlocality, or
	uncompletability after grouping parties
	\cite{Niset2006,Halder2019,Yuan2020,ShiHetero2022,ShiNPartite2022,
		ShiStrong2022,ZhouPlane2022,He2024,ShiNJP2022,Zhen2024}.  The difficulty is
	simultaneous: one and the same mutually orthogonal product family must remain
	unextendible for every bipartite grouping.  Existing work on the orthogonal
	problem has therefore been dominated by constraints and no-go results,
	including forbidden cardinalities \cite{Demianowicz2022}, graph-theoretic
	restrictions \cite{ShiGraph2023}, and the recent exclusion of the smallest
	three-qutrit candidate \cite{Demianowicz2026}.  No finite-dimensional
	orthogonal GUPB was previously known.
	
	The existence question has an immediate entanglement-theoretic consequence.
	For an orthogonal family of fully product vectors, partial transposition with
	respect to any subset of parties preserves the corresponding orthogonal
	projector.  Hence the normalized projector onto the complement of a GUPB is
	PPT with respect to every bipartition while being supported on a GES.  It is
	therefore genuinely multipartite bound entangled and lies inside the
	PPT-mixture relaxation, so every fully decomposable genuine multipartite entanglement (GME) witness is
	nonnegative on it \cite{Jungnitsch2011}; more general nondecomposable
	positive-map constructions show that PPT genuine multipartite entanglement
	can nevertheless be detected \cite{Huber2010,HuberSengupta2014}.  Thus an
	orthogonal GUPB would provide, from product vectors alone, a canonical sector
	of PPT genuine multipartite entanglement together with witnesses naturally
	adapted to that sector.
	
	The key technical ingredient of our solution is \emph{MDS minor rigidity}.
	A linear $[n,k,d]_p$ code is maximum distance separable (MDS) when it saturates
	the Singleton bound, $d=n-k+1$, equivalently when every $k$ columns of a
	full rank generator matrix are linearly independent \cite{Singleton1964};
	Reed--Solomon codes provide the standard explicit family
	\cite{ReedSolomon1960}.  MDS codes, orthogonal arrays, and related
	maximum-distance quantum codes have long been used in quantum error
	correction and in constructions of highly entangled states
	\cite{Scott2004,GourWallach2007,Goyeneche2014,Goyeneche2015,
		HuberGrassl2020}.  Here the same minor condition is used in a different way.
	For an $[N^2,N,N^2-N+1]_p$ generator matrix, transversal minors assign a unique tile to every computational cell, while the minors obtained by replacing one selected generator column control the linear transformations between neighboring fixed-address slices. Across an arbitrary bipartition, these transformations detect every displacement on the opposite subsystem and force any Cartesian union containing more than one tile to fill the entire computational grid. This bipartition-independent rectangle rigidity is the crucial structural consequence of the MDS property that enforces genuine unextendibility.
	
	We prove that every linear $[N^2,N,N^2-N+1]_p$ MDS code over a prime field
	produces a GUPB in $(\mathbb C^{Np})^{\otimes N}$ for every $N\geq3$;
	generalized Reed--Solomon codes therefore yield an explicit infinite family
	for all primes $p\geq N^2$.  On each tile we introduce local Fourier modes,
	delete one uniform mode, and add a global stopper. Product structure makes its
	support a Cartesian rectangle, and MDS rigidity leaves only a single tile or
	the full grid, both excluded by the stopper.  The complementary projector is
	invariant under partial transposition with respect to every bipartition, its
	canonical GME witness detects a full-rank PPT-GME family while being
	nondecomposable for every bipartition, and genuine unextendibility persists
	under arbitrary finite partywise tensor powers.  The construction therefore
	settles the orthogonal GUPB existence problem and identifies MDS minor
	rigidity as the key structural property underlying the construction.
	
	\medskip
	\noindent\emph{Construction GUPBs from MDS codes.---}
	We consider $N\geq3$ parties indexed by $P\in\Z_N$. The construction starts
	from a classical MDS code. Recall that a linear
	$[n,k,d]_p$ code $\cC\subseteq\F_p^n$ is MDS if it saturates the Singleton
	bound,
	\begin{equation}
		d=n-k+1.
		\label{eq:mds-singleton}
	\end{equation}
	Equivalently, if $\mathsf G\in\F_p^{k\times n}$ is a full-rank generator
	matrix of $\cC$, every $k$ columns of $\mathsf G$ are linearly independent
	\cite{Singleton1964}. We use
	\begin{equation}
		n=N^2,\qquad k=N,
		\label{eq:mds-parameters}
	\end{equation}
	so that $\cC$ has parameters $[N^2,N,N^2-N+1]_p$. The MDS condition will be
	used twice: first to partition the computational basis uniquely into product
	tiles, and then to establish the rigidity of such tiles with respect to every
	bipartition.
	
	Index the $N^2$ columns of the generator matrix by
	$(P,j)\in\Z_N\times\Z_N$,
	\begin{equation}
		\mathsf G
		=
		\bigl[
		\mathbf g_{0,0}\ \mathbf g_{0,1}\ \cdots\
		\mathbf g_{N-1,N-1}
		\bigr]
		\in\F_p^{N\times N^2},
		\label{eq:generator-matrix}
	\end{equation}
	where $\mathbf g_{P,j}\in\F_p^N$. For a message $\mathbf t\in\F_p^N$, let
	\begin{equation}
		\mathbf c(\mathbf t)
		=
		\mathbf t^{\transpose}\mathsf G
		\in\cC,
		\qquad
		c_{P,j}(\mathbf t)
		=
		\mathbf t^{\transpose}\mathbf g_{P,j}
		\label{eq:codeword-coordinate}
	\end{equation}
	be the corresponding codeword and its $(P,j)$th coordinate.
	
	At party $P\in\Z_N$, take
	\begin{equation}
		\cH_P
		=
		\Span\{\ket{j,r}_P:j\in\Z_N,\ r\in\F_p\}
		\cong\C^{Np}.
		\label{eq:local-space}
	\end{equation}
	Each message $\mathbf t$ determines a local set
	\begin{equation}
		F_P(\mathbf t)
		=
		\{(j,c_{P,j}(\mathbf t)):j\in\Z_N\}
		\label{eq:local-fiber}
	\end{equation}
	and hence a product tile
	\begin{equation}
		T_{\mathbf t}
		=
		\prod_{P\in\Z_N}F_P(\mathbf t).
		\label{eq:tile}
	\end{equation}
	For later use, if $Z\subseteq\Z_N$ is a set of parties, write $
	F_Z(\mathbf t)=\prod_{P\in Z}F_P(\mathbf t).
	$
	The $p^N$ tiles $T_{\mathbf t}$ partition the full computational grid
	$(\Z_N\times\F_p)^N$. Indeed, a computational cell
	$((j_P,r_P))_{P\in\Z_N}$ belongs to $T_{\mathbf t}$ precisely when
	\begin{equation}
		c_{P,j_P}(\mathbf t)=r_P,
		\qquad P\in\Z_N.
		\label{eq:tile-linear-system}
	\end{equation}
	This is a system of $N$ linear equations for $\mathbf t$. Its coefficient
	matrix consists of the transposes of the $N$ selected columns
	$\mathbf g_{P,j_P}$, which are linearly independent by the MDS property.
	Hence Eq.~\eqref{eq:tile-linear-system} has a unique solution $\mathbf t$.
	
	We now place an orthogonal product basis on each tile. Let
	$\zeta=e^{2\pi i/N}$ and define
	\begin{equation}
		\ket{f_{P,\mathbf t}(a)}
		=
		\sum_{j=0}^{N-1}
		\zeta^{aj}
		\ket{j,c_{P,j}(\mathbf t)}_P,
		\qquad
		a\in\Z_N.
		\label{eq:local-fourier}
	\end{equation}
	For fixed $\mathbf t$,
	\begin{equation}
		\cB_{\mathbf t}
		=
		\left\{
		\bigotimes_P\ket{f_{P,\mathbf t}(a_P)}
		:
		a_P\in\Z_N
		\right\} \subseteq   \bigotimes_{P\in \mathbb{Z}_N}\mathcal{H}_P.
		\label{eq:tile-basis}
	\end{equation}
	Delete the all zero Fourier mode
	\begin{equation}
		\ket{\psi_{\mathbf t}}
		=
		\bigotimes_{P\in\Z_N}\ket{f_{P,\mathbf t}(0)}
		\label{eq:tile-hole}
	\end{equation}
	and add the stopper
	\begin{equation}
		\ket S
		=
		\bigotimes_{P\in\Z_N}
		\left(
		\sum_{j\in\Z_N}\sum_{r\in\F_p}\ket{j,r}_P
		\right).
		\label{eq:stopper}
	\end{equation}
	Our candidate is
	\begin{equation}
		\cU_{\cC}
		=
		\bigcup_{\mathbf t\in\F_p^N}
		\left(
		\cB_{\mathbf t}\setminus\{\ket{\psi_{\mathbf t}}\}
		\right)
		\cup\{\ket S\}.
		\label{eq:gupb-set}
	\end{equation}
	Normalization factors are omitted throughout. Different tiles have disjoint
	computational support, while every retained vector in $\cB_{\mathbf t}$ has
	at least one nonzero Fourier label and is therefore orthogonal to $\ket S$.
	Thus $\cU_{\cC}$ is an orthogonal family of fully product states, with
	\begin{equation}
		\abs{\cU_{\cC}}
		=
		p^N(N^N-1)+1,
		\qquad
		\dim\Span(\cU_{\cC})^\perp
		=
		p^N-1.
		\label{eq:size-count}
	\end{equation}
	
	The main result is the following.
	
	\begin{theorem}[GUPBs from MDS codes]
		\label{thm:gupb-existence}
		Let $N\geq3$, let $p$ be prime, and let $\cC$ be a linear
		$[N^2,N,N^2-N+1]_p$ MDS code. Then the family $\cU_{\cC}$ defined in
		Eq.~\eqref{eq:gupb-set} is a genuinely unextendible product basis in
		$(\C^{Np})^{\otimes N}$.
	\end{theorem}
	
	First, observe that
	\begin{equation}
		\left[
		\bigcup_{\mathbf t\in\F_p^N}
		\left(
		\cB_{\mathbf t}\setminus\{\ket{\psi_{\mathbf t}}\}
		\right)
		\right]^\perp
		=
		\Span\left\{
		\ket{\psi_{\mathbf t}}:
		\mathbf t\in\F_p^N
		\right\},
		\label{eq:deleted-mode-span}
	\end{equation}
	and, with the normalization factors omitted as above,
	\begin{equation}
		\ket S
		=
		\sum_{\mathbf t\in\F_p^N}
		\ket{\psi_{\mathbf t}}.
		\label{eq:stopper-sum}
	\end{equation}
	Therefore every nonzero vector
	$\ket{\psi}\in\cU_{\cC}^{\perp}$ can be written uniquely as
	\begin{equation}
		\ket{\psi}
		=
		\sum_{\mathbf t\in\F_p^N}
		a_{\mathbf t}\ket{\psi_{\mathbf t}},
		\qquad
		\sum_{\mathbf t\in\F_p^N}a_{\mathbf t}=0.
		\label{eq:complement-expansion-proof}
	\end{equation}
	Define the set of active tiles by
	\begin{equation}
		\Lambda(\psi)
		=
		\left\{
		\mathbf t\in\F_p^N:
		a_{\mathbf t}\neq0
		\right\}.
		\label{eq:active-set-proof}
	\end{equation}
	Since $\ket{\psi}\neq0$ and the coefficients in
	Eq.~\eqref{eq:complement-expansion-proof} sum to zero,
	\begin{equation}
		|\Lambda(\psi)|\geq2.
		\label{eq:at-least-two-active-proof}
	\end{equation}
	
	Suppose now that $\ket{\psi}$ is product across a nontrivial bipartition
	$X|Y$,
	\begin{equation}
		\ket{\psi}
		=
		\ket{\alpha}_X\otimes\ket{\beta}_Y.
		\label{eq:product-assumption-proof}
	\end{equation}
	For any state $\ket{\phi}$, let
	$\operatorname{supp}_{\rm comp}(\phi)$ denote the set of computational
	basis cells on which $\ket{\phi}$ has a nonzero coefficient. Since each
	$\ket{\psi_{\mathbf t}}$ has nonzero coefficients on every cell of
	$T_{\mathbf t}$ and vanishes outside $T_{\mathbf t}$, while different
	tiles have disjoint support, Eq.~\eqref{eq:complement-expansion-proof}
	gives
	\begin{equation}
		\operatorname{supp}_{\rm comp}(\psi)
		=
		\bigcup_{\mathbf t\in\Lambda(\psi)}
		T_{\mathbf t}.
		\label{eq:support-union-proof}
	\end{equation}
	On the other hand, the support of a bipartite product vector factorizes:
	\begin{equation}
		\operatorname{supp}_{\rm comp}(\psi)
		=
		\operatorname{supp}_{\rm comp}(\alpha)
		\times
		\operatorname{supp}_{\rm comp}(\beta).
		\label{eq:support-product-proof}
	\end{equation}
	Indeed, the coefficient of the computational basis vector
	$\ket{q_X,q_Y}$ is
	$\alpha(q_X)\beta(q_Y)$, which is nonzero if and only if both factors
	are nonzero. Hence
	$\bigcup_{\mathbf t\in\Lambda(\psi)}T_{\mathbf t}$ is an
	$X|Y$ Cartesian rectangle.
	
	The crucial input is the following consequence of the MDS property.
	
	\begin{lemma}[MDS rectangle rigidity]
		\label{lem:rectangle-rigidity}
		Let $\Lambda\subseteq\F_p^N$ with $|\Lambda|\geq2$. If
		$
		\bigcup_{\mathbf t\in\Lambda}T_{\mathbf t}
		$
		is an $X|Y$ Cartesian rectangle, then
		\begin{equation}
			\Lambda=\F_p^N.
		\end{equation}
	\end{lemma}
	
	The proof, which relies essentially on the nonvanishing minors of the MDS
	generator matrix, is given in Supplemental
	Material.
	
	Applying Lemma~\ref{lem:rectangle-rigidity} to
	Eqs.~\eqref{eq:support-union-proof} and
	\eqref{eq:support-product-proof}, together with
	$|\Lambda(\psi)|\geq2$, yields
	\begin{equation}
		\Lambda(\psi)=\F_p^N.
		\label{eq:all-tiles-active-proof}
	\end{equation}
	Thus every tile is active and every coefficient $a_{\mathbf t}$ is
	nonzero.
	
	For $q_X$ and $q_Y$ in the computational bases of $X$ and $Y$, respectively,
	write
	\begin{equation}
		\alpha(q_X)=\langle q_X|\alpha\rangle_X,
		\qquad
		\beta(q_Y)=\langle q_Y|\beta\rangle_Y.
	\end{equation}
	Because the tiles partition the computational grid, for every pair
	$(q_X,q_Y)$ there is a unique $\mathbf t\in\F_p^N$ such that
	\begin{equation}
		(q_X,q_Y)
		\in
		F_X(\mathbf t)\times F_Y(\mathbf t)
		=
		T_{\mathbf t}.
	\end{equation}
	On this tile, Eq.~\eqref{eq:complement-expansion-proof} gives
	\begin{equation}
		\alpha(q_X)\beta(q_Y)
		=
		(\langle q_X|\otimes\langle q_Y|)\ket{\psi}
		=
		a_{\mathbf t}.
		\label{eq:coefficient-on-tile-proof}
	\end{equation}
	Since $a_{\mathbf t}\neq0$, both
	$\alpha(q_X)$ and $\beta(q_Y)$ are nonzero.
	
	Now fix $\mathbf t$ and choose
	$q_Y\in F_Y(\mathbf t)$. For any
	$q_X,q_X'\in F_X(\mathbf t)$, both
	$(q_X,q_Y)$ and $(q_X',q_Y)$ belong to $T_{\mathbf t}$.
	Hence Eq.~\eqref{eq:coefficient-on-tile-proof} gives
	\begin{equation}
		\alpha(q_X)\beta(q_Y)
		=
		a_{\mathbf t}
		=
		\alpha(q_X')\beta(q_Y).
		\label{eq:proof-constant-on-tile}
	\end{equation}
	Because $\beta(q_Y)\neq0$, we obtain
	\begin{equation}
		\alpha(q_X)=\alpha(q_X').
	\end{equation}
	Thus $\alpha$ is constant on every  
	$F_X(\mathbf t)$. By the same argument,
	$\beta$ is constant on  every  
	$F_Y(\mathbf t).$

	As proved in the Lemma \ref{lem:incidence} in Supplemental material, the incidence graphs generated by
	the projected tiles $F_X(\mathbf t)$ and $F_Y(\mathbf t)$ are connected.
	The local constants therefore agree throughout the full computational
	bases of $X$ and $Y$. Consequently,
	\begin{equation}
		\ket{\alpha}_X
		\propto
		\sum_{q_X}\ket{q_X},
		\qquad
		\ket{\beta}_Y
		\propto
		\sum_{q_Y}\ket{q_Y}.
		\label{eq:uniform-factors-proof}
	\end{equation}
	It follows that
	\begin{equation}
		\ket{\psi}
		=
		\ket{\alpha}_X\otimes\ket{\beta}_Y
		\propto
		\ket S.
	\end{equation}
	This is impossible because
	$\ket{\psi}\in\cU_{\cC}^{\perp}$ whereas
	$\ket S\in\cU_{\cC}$.
	
	Hence no nonzero product vector across $X|Y$ is orthogonal to
	$\cU_{\cC}$. Since the bipartition $X|Y$ was arbitrary,
	$\cU_{\cC}$ is a GUPB.

	\vskip 10pt
	
	For an explicit family, choose $N^2$ distinct elements
	$\alpha_{P,j}\in\F_p$ and nonzero multipliers $v_{P,j}\in\F_p$, possible
	whenever $p\geq N^2$, and take
	\begin{equation}
		\mathbf g_{P,j}
		=
		v_{P,j}
		\begin{pmatrix}
			1&\alpha_{P,j}&\alpha_{P,j}^{2}&\cdots&\alpha_{P,j}^{N-1}
		\end{pmatrix}^{\transpose}.
		\label{eq:grs-realization}
	\end{equation}
	These columns generate a generalized Reed--Solomon
	$[N^2,N,N^2-N+1]_p$ code \cite{ReedSolomon1960}. Hence
	Theorem~\ref{thm:gupb-existence} gives an explicit GUPB for every $N\geq3$
	and every prime $p\geq N^2$.

	\medskip
	\noindent\emph{Genuine bound entanglement and its witness.---}
	We have showed that
	\begin{equation}
		\cG_{\cC}
		=
		\Span(\cU_{\cC})^\perp
		=
		\left\{
		\sum_{\mathbf t\in\F_p^N}
		a_{\mathbf t}\ket{\psi_{\mathbf t}}
		:
		\sum_{\mathbf t}a_{\mathbf t}=0
		\right\}.
		\label{eq:complement}
	\end{equation}
	Let $\Pi_{\cG}$ be the projector onto $\cG_{\cC}$ and define
	\begin{equation}
		\begin{aligned}
			P_{\cU}&=I-\Pi_{\cG},
			&
			R&=p^N-1,\\
			K&=(Np)^N-p^N+1,
			&
			D&=(Np)^N.
		\end{aligned}
		\label{eq:dimensions}
	\end{equation}
	The normalized complementary state is
	\begin{equation}
		\rho_{\cG}
		=
		\frac{\Pi_{\cG}}{R}.
		\label{eq:complement-state}
	\end{equation}
	The partial-transpose symmetry is most transparent from the normalized deleted
	modes.  Let
	\begin{equation}
		\ket{e_{\mathbf t}}=N^{-N/2}\ket{\psi_{\mathbf t}},
		\qquad
		\ket s=D^{-1/2}\ket S
		=p^{-N/2}\sum_{\mathbf t\in\F_p^N}\ket{e_{\mathbf t}}.
		\label{eq:normalized-hole-stopper-main}
	\end{equation}
	As shown in Supplemental Sec.~\ref{sec:supp-witness},
	\begin{equation}
		\Pi_{\cG}
		=
		\sum_{\mathbf t\in\F_p^N}\proj{e_{\mathbf t}}-\proj{s}.
		\label{eq:complement-projector-main}
	\end{equation}
	Every \(\ket{e_{\mathbf t}}\) and \(\ket s\) is a fully product vector with
	real coefficients in the computational basis.  Hence each rank-one projector
	in Eq.~\eqref{eq:complement-projector-main} is fixed by partial transposition
	on an arbitrary subset of parties.  Therefore, for every
	$X\subseteq\Z_N$,
	\begin{equation}
		\Pi_{\cG}^{T_X}=\Pi_{\cG},
		\qquad P_{\cU}^{T_X}=P_{\cU},
		\qquad
		\rho_{\cG}^{T_X}=\rho_{\cG}.
		\label{eq:pt-invariance}
	\end{equation}
	Theorem~\ref{thm:gupb-existence} makes $\rho_{\cG}$ a   GME, whereas
	Eq.~\eqref{eq:pt-invariance} makes it PPT, and therefore nondistillable,
	across every bipartition \cite{Horodecki1998}.
	
	Let $\BSEP$ denote the convex set of normalized biseparable states and
	define
	\begin{equation}
		\epsilon_G
		=
		\min_{\sigma\in\BSEP}
		\Tr(P_{\cU}\sigma)>0.
		\label{eq:biseparable-threshold}
	\end{equation}
	Then
	\begin{equation}
		W_G=P_{\cU}-\epsilon_G I
		\label{eq:gme-witness}
	\end{equation}
	is a GME witness and
	$\Tr(W_G\rho_{\cG})=-\epsilon_G$.
	The witness also separates the GUPB state from the fully decomposable class.
	
	\begin{theorem}[GME states with PPT across every bipartition]
		\label{thm:ppt-gme-witness}
		For every MDS-code GUPB of Theorem~\ref{thm:gupb-existence}, $W_G$ is
		nondecomposable with respect to every bipartition. Moreover, for
		\begin{equation}
			\rho_\lambda
			=
			\frac{\Pi_{\cG}+\lambda P_{\cU}}
			{R+\lambda K},
			\qquad 0<\lambda\leq1,
			\label{eq:strict-ppt-family}
		\end{equation}
		one has $\rho_\lambda^{T_X}=\rho_\lambda>0$ for every nontrivial
		bipartition $X|Y$. The state is GME whenever
		\begin{equation}
			0<\lambda<
			\kappa_G
			:=
			\frac{R\epsilon_G}{K(1-\epsilon_G)}.
			\label{eq:ppt-gme-threshold}
		\end{equation}
		No fully decomposable GME witness detects any member of this family, whereas
		$W_G$ detects the interval in Eq.~\eqref{eq:ppt-gme-threshold}.
	\end{theorem}
	
	Indeed,
	\begin{equation}
		\Tr(W_G\rho_\lambda)
		=
		\frac{\lambda K}{R+\lambda K}-\epsilon_G,
		\label{eq:strict-ppt-witness-value}
	\end{equation}
	which is negative exactly in the stated interval. If $W_G$ were decomposable
	with respect to one bipartition, $W_G=A_X+B_X^{T_X}$ with $A_X,B_X\geq0$, then its expectation
	on the PPT state $\rho_{\cG}$ would be nonnegative, contrary to
	$\Tr(W_G\rho_{\cG})=-\epsilon_G$. Thus the same witness is nondecomposable with respect to every bipartition. Since every $\rho_\lambda$ is PPT across every bipartition, it lies in
	the PPT-mixture set and all fully decomposable witnesses are nonnegative on
	it (the detail proof is presented at Supplemental Material).
	
	\medskip
	\noindent\emph{Measurement and discrimination consequences.---}
	The same complementary projector has a direct measurement interpretation;
	the proofs of the statements below are collected in Supplemental
	Sec.~\ref{sec:supp-measurement}.
	Equation~\eqref{eq:pt-invariance} makes
	\begin{equation}
		\mathsf M_G=\{P_{\cU},\Pi_{\cG}\}
		\label{eq:binary-gupb-measurement}
	\end{equation}
	PPT across every bipartition. Nevertheless
	$\Pi_{\cG}$ is not a biseparable positive operator. Otherwise, a nonzero separable summand in a biseparable decomposition would
	have a product vector, with respect to its defining bipartition, in its range, which would also lie in $\cG_{\cC}$, contradicting
	Theorem~\ref{thm:gupb-existence}. Thus
	\begin{equation}
		\Pi_{\cG}
		\in
		\left(
		\bigcap_{X|Y}\PPT_{X|Y}^{+}
		\right)
		\setminus\BSEP^{+}.
		\label{eq:ppt-bsep-separation}
	\end{equation}
	This observation gives a simple discrimination statement. Put
	\begin{equation}
		\rho_{\cU}=\frac{P_{\cU}}{K},
		\qquad
		\rho_{\cG}=\frac{\Pi_{\cG}}{R}.
		\label{eq:support-complement-states}
	\end{equation}
	For every element $E\in\BSEP^+$, let
	$q_{\cU}=\Tr(E\rho_{\cU})$ and $q_{\cG}=\Tr(E\rho_{\cG})$. Then
	\begin{equation}
		q_{\cU}\geq\kappa_G q_{\cG}.
		\label{eq:one-copy-biseparable-bound}
	\end{equation}
	A derivation is given in the Supplemental Material. The constant is sharp. In particular, no effect $E\in\BSEP^+$ can satisfy
	$\Tr(E\rho_{\cU})=0$ and $\Tr(E\rho_{\cG})>0$. By contrast, the projective
	measurement $\{P_{\cU},\Pi_{\cG}\}$ distinguishes the two states perfectly,
	and both effects are PPT across every bipartition.

	\medskip
	\noindent\emph{Finite copy  local discrimination.---} The following show a gap between finte copy biseparate measurement and PPT measurement.
	
	\begin{theorem}[Finite-copy  local discrimination]
		\label{thm:tensor-separation}
		For every $\ell\geq1$, the $\ell$-fold tensor power
		$\cU_{\cC}^{\otimes\ell}$, with the copies grouped according to the original
		parties, is a GUPB. Fix a bipartition $X|Y$ and normalized vectors $\ket A\in\cH_X^{\otimes\ell}$
		and $\ket B\in\cH_Y^{\otimes\ell}$. Define
		\begin{align}
			p_{\cU}^{(\ell)}(A,B)
			&=
			\bra{A,B}P_{\cU}^{\otimes\ell}\ket{A,B},
			\label{eq:tensor-p}\\
			p_{\cG}^{(\ell)}(A,B)
			&=
			\bra{A,B}\Pi_{\cG}^{\otimes\ell}\ket{A,B},
			\label{eq:tensor-g}\\
			\kappa_{\ell,X|Y}
			&=
			\left(\frac{R}{K}\right)^\ell
			\min_{\substack{\norm A=\norm B=1\\p_{\cG}^{(\ell)}(A,B)>0}}
			\frac{p_{\cU}^{(\ell)}(A,B)}{p_{\cG}^{(\ell)}(A,B)}.
			\label{eq:tensor-slope-definition}
		\end{align}
		Then
		\begin{equation}
			0<\kappa_{\ell,X|Y}\leq1,
			\label{eq:kappa-range}
		\end{equation}
		and every $X|Y$-separable effect $E$ satisfies
		\begin{equation}
			\Tr(E\rho_{\cU}^{\otimes\ell})
			\geq
			\kappa_{\ell,X|Y}
			\Tr(E\rho_{\cG}^{\otimes\ell}).
			\label{eq:tensor-slope}
		\end{equation}
		The constant is sharp. By contrast, the measurement
		\begin{equation}
			\left\{
			P_{\cU}^{\otimes\ell},
			I-P_{\cU}^{\otimes\ell}
			\right\}
			\label{eq:tensor-ppt-measurement}
		\end{equation}
		is invariant under partial transposition across every bipartition and
		perfectly distinguishes $\rho_{\cU}^{\otimes\ell}$ from
		$\rho_{\cG}^{\otimes\ell}$.
	\end{theorem}
	
	The tensor closure is the usual bipartite UPB tensor-product closure applied
	to every bipartition \cite{DiVincenzo2003}. A complete proof of the closure,
	the finite-copy bound, and its optimality is given in Supplemental
	Material. That section also gives
	\begin{equation}
		\kappa_{1,X|Y}
		=
		\frac{R\epsilon_{X|Y}}
		{K(1-\epsilon_{X|Y})},
		\qquad
		\epsilon_{X|Y}
		=
		\min_{\norm a=\norm b=1}
		\bra{a,b}P_{\cU}\ket{a,b},
		\label{eq:one-copy-slope}
	\end{equation}
	and
	\begin{equation}
		\kappa_{\ell+r,X|Y}
		\leq
		\kappa_{\ell,X|Y}
		\kappa_{r,X|Y}.
		\label{eq:slope-submultiplicative}
	\end{equation}
	Consequently the limit
	\begin{equation}
		\xi_{X|Y}
		=
		-\lim_{\ell\to\infty}
		\frac{1}{\ell}\log\kappa_{\ell,X|Y}
		\label{eq:slope-exponent}
	\end{equation}
	exists, possibly as $+\infty$. The statement holds for every finite copy number: arbitrary joint operations
	on all copies are allowed within each side of the bipartition, but  still a
	separable-measurement. 
	
	\medskip
	\noindent\emph{Discussion.---}
	The result settles the existence problem for orthogonal GUPBs without
	relying on a small exceptional example. The construction applies to every
	linear $[N^2,N,N^2-N+1]_p$ MDS code over a prime field. The proof is also transparent: transversal minors assign each computational cell to one
	tile, the remaining MDS minors force rectangle rigidity across every bipartition,
	Fourier deletion leaves one deleted mode per tile, and the stopper removes the two
	rectangular cases that survive the rigidity lemma.
	
	The complementary projector also gives the main physical consequences of the construction. Its range is a GES, while the projector is invariant under partial transposition across every bipartition. The associated witness is therefore
	nondecomposable with respect to every bipartition and detects full-rank states lying in
	the interior of every PPT cone, a regime inaccessible to fully decomposable
	GME witnesses. The same projector pair gives a measurement effect that is PPT across every
	bipartition but lies outside the biseparable cone. It also yields, at every
	finite copy number, a separation between a measurement invariant under
	partial transposition and measurements separable across a fixed bipartition.
	
	Bertrand's postulate allows $N^2\leq p<2N^2$, so the local dimension can be
	chosen below $2N^3$. The present family is not expected to be dimension
	optimal. Smaller GUPBs, tighter values of $\epsilon_G$, the prime-power
	extension without an additional additive-mixing hypothesis, and the asymptotic
	exponent in Eq.~\eqref{eq:slope-exponent} remain open. 
	
	\medskip
	\noindent\emph{Acknowledgments.---}
	This work was supported by the Guangdong Basic and Applied Basic Research
	Foundation under Grant No.~2024A1515010380 and by the National Natural
	Science Foundation of China under Grant No.~12371458.
	
	\bibliographystyle{apsrev4-1}
	\bibliography{GUPB_MDS_references}
	
	\clearpage
	
	%%%%%%%%%%%%%%%%%%%%%%%%%%%%
	\onecolumngrid
	\appendix
	
	\setcounter{equation}{0}
	\renewcommand{\theequation}{S\arabic{equation}}
	\renewcommand{\theHequation}{S\arabic{equation}}
	\setcounter{theorem}{0}
	\renewcommand{\thetheorem}{S\arabic{theorem}}
	\renewcommand{\theHtheorem}{S\arabic{theorem}}
	\setcounter{section}{0}
	\renewcommand{\thesection}{S\Roman{section}}
	\renewcommand{\theHsection}{S\Roman{section}}
	\renewcommand{\thesubsection}{\thesection.\Alph{subsection}}
	\renewcommand{\theHsubsection}{\theHsection.\Alph{subsection}}
	
	\begin{center}
		{\large\bfseries Supplemental Material for\\[0.25em]
			``Genuinely Unextendible Product Bases from Maximum Distance Separable Codes''\par}
		\vspace{0.7em}
		{\large Mao-Sheng Li\par}
	\end{center}
	\vspace{0.5em}
	
	\noindent\textbf{Common notation.---} The notation is identical to that of the Letter.  In particular,
	$\mathsf G$ is the generator matrix, $\mathbf g_{P,j}$ is its $(P,j)$th
	column, $\mathbf t\in\F_p^N$ is a message and tile label,
	$\mathbf c(\mathbf t)=\mathbf t^{\transpose}\mathsf G$ is the associated
	codeword, and $c_{P,j}(\mathbf t)$ is its $(P,j)$th coordinate.  Quantum
	subspaces are denoted by calligraphic letters, projectors by $P$ or $\Pi$,
	and coefficients in superpositions are never denoted by $c$, so that they
	cannot be confused with codeword coordinates.
	
	Let $N\geq3$ be an integer, let $p$ be a prime, let
	$\Z_N=\{0,1,\ldots,N-1\}$ denote the integers modulo $N$, and let $\F_p$
	denote the field with $p$ elements. Unless stated otherwise, vectors in
	$\F_p^N$ are columns. Throughout, $\cC\subseteq\F_p^{N^2}$ is an arbitrary
	linear $[N^2,N,N^2-N+1]_p$ maximum-distance-separable (MDS) code, and
	\begin{equation}
		\mathsf G
		=
		\bigl[\mathbf g_{0,0}\ \mathbf g_{0,1}\ \cdots\
		\mathbf g_{N-1,N-1}\bigr]
		\in\F_p^{N\times N^2}
		\label{eq:supp-generator-matrix}
	\end{equation}
	is a fixed full rank generator matrix. Its columns are labeled by
	$(P,j)\in\Z_N\times\Z_N$. For $\mathbf t\in\F_p^N$ we use the column-vector convention for messages and write the corresponding codeword as the row
	vector
	\begin{equation}
		\mathbf c(\mathbf t)
		=
		\mathbf t^{\transpose}\mathsf G,
		\qquad
		c_{P,j}(\mathbf t)
		=
		\mathbf t^{\transpose}\mathbf g_{P,j}.
		\label{eq:supp-codeword-coordinate}
	\end{equation}

	The party labeled $P\in\Z_N$ has Hilbert space
	$\cH_P=\Span\{\ket{j,r}_P:j\in\Z_N, r\in\F_p\}\cong\C^{Np}$, with the
	displayed vectors as an orthonormal computational basis. The total Hilbert
	space is $\cH=\bigotimes_{P\in\Z_N}\cH_P$. A \emph{computational cell} is
	an $N$-tuple $\mathbf q=((j_P,r_P))_{P\in\Z_N}$ labeling the basis vector
	$\ket{\mathbf q}=\bigotimes_P\ket{j_P,r_P}_P$. For a vector $w\in\cH$,
	$\suppcomp(w)$ is the set of cells on which $w$ has nonzero computational
	coefficient.
	
	For a message $\mathbf t\in\F_p^N$, define the local fiber
	$F_P(\mathbf t)=\{(j,c_{P,j}(\mathbf t)):j\in\Z_N\}$ and the
	tile $T_{\mathbf t}=\prod_{P\in\Z_N}F_P(\mathbf t)$. Let
	$\cH_{\mathbf t}=\Span\{\ket{\mathbf q}:\mathbf q\in T_{\mathbf t}\}$
	and let $\Pi_{\mathbf t}$ be the orthogonal projector onto
	$\cH_{\mathbf t}$. With $\zeta=e^{2\pi i/N}$ and $a\in\Z_N$, set
	$\ket{f_{P,\mathbf t}(a)}=\sum_{j=0}^{N-1}\zeta^{aj}
	\ket{j,c_{P,j}(\mathbf t)}_P$. The product Fourier basis of a
	tile is $\cB_{\mathbf t}=\{\bigotimes_P
	\ket{f_{P,\mathbf t}(a_P)}:a_P\in\Z_N\}$, and its deleted vector is
	$\ket{\psi_{\mathbf t}}=\bigotimes_P\ket{f_{P,\mathbf t}(0)}$. The global
	stopper is $\ket S=\bigotimes_P\sum_{j\in\Z_N,r\in\F_p}\ket{j,r}_P$, and
	the constructed orthogonal family is
	$\cU_{\cC}=\bigcup_{\mathbf t}(\cB_{\mathbf t}\setminus
	\{\ket{\psi_{\mathbf t}}\})\cup\{\ket S\}$.
	
	A nontrivial bipartition $X|Y$ is a disjoint decomposition
	$\Z_N=X\mathbin{\dot\cup}Y$ with $X,Y\neq\varnothing$. For
	$Z\in\{X,Y\}$, write $\cH_Z=\bigotimes_{P\in Z}\cH_P$ and
	$F_Z(\mathbf t)=\prod_{P\in Z}F_P(\mathbf t)$. An $X|Y$ product vector,
	also called a product vector across a bipartition, is
	$\ket{\alpha}_X\otimes\ket{\beta}_Y\in\cH_X\otimes\cH_Y$. A genuinely
	unextendible product basis (GUPB) is an incomplete orthogonal family of fully
	product vectors whose orthogonal complement contains no nonzero product vector
	with respect to any nontrivial bipartition.

	\section{Construction checks}
	\label{sec:supp-construction}

	\begin{proposition}[Tile partition and orthogonality]
		\label{prop:supp-construction-checks}
		The sets $\{T_{\mathbf t}:\mathbf t\in\F_p^N\}$ form a partition of
		$(\Z_N\times\F_p)^N$.  For each $\mathbf t$, $\cB_{\mathbf t}$ is an
		orthogonal product basis of the subspace supported on $T_{\mathbf t}$.
		Moreover, the family $\cU_{\cC}$ in Eq.~\eqref{eq:gupb-set} is orthogonal
		and has the cardinality and complementary dimension in
		Eq.~\eqref{eq:size-count}.
	\end{proposition}
	
	\begin{proof}
		Fix a computational cell
		$\mathbf q=((j_P,r_P))_{P\in\Z_N}$.  The condition
		$\mathbf q\in T_{\mathbf t}$ is equivalent to
		\begin{equation}
			\mathbf t^{\transpose}\mathbf g_{P,j_P}=r_P,
			\qquad P\in\Z_N.
			\label{eq:supp-tile-system}
		\end{equation}
		The coefficient matrix of this system has rows
		$\mathbf g_{P,j_P}^{\transpose}$.  The $N$ selected generator columns are
		distinct and hence independent; the
		system therefore has a unique solution $\mathbf t$.  Thus every
		computational cell belongs to exactly one tile.
		
		For fixed $P$ and $\mathbf t$,
		\begin{equation}
			\braket{f_{P,\mathbf t}(a)}{f_{P,\mathbf t}(b)}
			=
			\sum_{j=0}^{N-1}\zeta^{(b-a)j}
			=
			N\,\delta_{a,b}.
			\label{eq:supp-local-fourier-orthogonality}
		\end{equation}
		Hence the $N^N$ product vectors in $\cB_{\mathbf t}$ are mutually
		orthogonal.  They all lie in the $N^N$-dimensional tile subspace, so they
		form a basis.  Distinct tile bases are mutually orthogonal because their
		computational supports are disjoint.
		
		Finally, let
		$\ket{s_P}=\sum_{j\in\Z_N,r\in\F_p}\ket{j,r}_P$.  Then
		\begin{equation}
			\braket{s_P}{f_{P,\mathbf t}(a)}
			=
			\sum_{j=0}^{N-1}\zeta^{aj}
			=
			N\,\delta_{a,0}.
			\label{eq:supp-stopper-local-overlap}
		\end{equation}
		Every retained tile vector has at least one nonzero Fourier label and is
		therefore orthogonal to the global stopper.  Consequently $\cU_{\cC}$ is
		orthogonal.  It contains $p^N(N^N-1)$ retained tile vectors and one stopper,
		so
		\begin{equation}
			\abs{\cU_{\cC}}=p^N(N^N-1)+1.
		\end{equation}
		The total Hilbert-space dimension is $(Np)^N=p^NN^N$, which gives
		$\dim\Span(\cU_{\cC})^\perp=p^N-1$.
	\end{proof}

	\section{MDS rectangle rigidity and projected-tile connectivity}
	\label{sec:supp-rectangle-rigidity}

	Fix a nontrivial bipartition $X|Y$ throughout. For an address word
	$\mathbf j=(j_P)_{P\in\Z_N}\in\Z_N^N$, define the fixed-address slice
	\begin{equation}
		\Sigma_{\mathbf j}
		=
		\left\{
		((j_P,r_P))_{P\in\Z_N}:(r_P)_P\in\F_p^N
		\right\}
		\label{eq:supp-address-slice}
	\end{equation}
	and the transversal matrix
	\begin{equation}
		M_{\mathbf j}
		=
		\begin{pmatrix}
			\mathbf g_{0,j_0}^{\transpose}\\
			\vdots\\
			\mathbf g_{N-1,j_{N-1}}^{\transpose}
		\end{pmatrix}.
		\label{eq:supp-transversal-matrix}
	\end{equation}
	The selected columns of $\mathsf G$ are distinct, so 
	\begin{equation}
		\det M_{\mathbf j}\neq0.
		\label{eq:supp-transversal-invertible}
	\end{equation}
	
	\begin{lemma}[Slice image]
		\label{lem:supp-slice-image}
		Under the identification $\Sigma_{\mathbf j}\cong\F_p^N$ by field-symbol
		coordinates,
		\begin{equation}
			T_{\mathbf t}\cap\Sigma_{\mathbf j}
			\longleftrightarrow
			M_{\mathbf j}\mathbf t.
			\label{eq:supp-single-tile-slice}
		\end{equation}
		Consequently,
		\begin{equation}
			\left(\bigcup_{\mathbf t\in\Lambda}T_{\mathbf t}\right)
			\cap\Sigma_{\mathbf j}
			\longleftrightarrow
			M_{\mathbf j}\Lambda.
			\label{eq:supp-slice-image}
		\end{equation}
		If the tiled union is an $X|Y$ Cartesian rectangle, then for every
		$\mathbf j$ there are sets $A_{\mathbf j}\subseteq\F_p^X$ and
		$B_{\mathbf j}\subseteq\F_p^Y$ such that
		\begin{equation}
			M_{\mathbf j}\Lambda=A_{\mathbf j}\times B_{\mathbf j}.
			\label{eq:supp-slice-cartesian}
		\end{equation}
	\end{lemma}
	
	\begin{proof}
		For fixed $\mathbf j$ and $\mathbf t$, the unique point of $T_{\mathbf t}$
		in $\Sigma_{\mathbf j}$ has field-symbol vector
		\begin{equation}
			\bigl(c_{P,j_P}(\mathbf t)\bigr)_{P\in\Z_N}
			=
			M_{\mathbf j}\mathbf t.
		\end{equation}
		This proves the first two statements.  Intersecting a Cartesian set
		$L_X\times L_Y$ with a fixed-address slice fixes only the address entries
		and therefore preserves the Cartesian product between the $X$ and $Y$
		field-symbol coordinates, proving Eq.~\eqref{eq:supp-slice-cartesian}.
	\end{proof}
	
	Fix a baseline address word $\mathbf j$ and write $M=M_{\mathbf j}$.  For
	$P\in\Z_N$ and an alternative address $k\neq j_P$, define the
	replacement-coordinate row
	\begin{equation}
		\boldsymbol{\gamma}^{(P,k)}
		=
		\mathbf g_{P,k}^{\transpose}M^{-1}
		=
		\bigl(\gamma_Q^{(P,k)}\bigr)_{Q\in\Z_N}.
		\label{eq:supp-coordinate-row}
	\end{equation}
	
	\begin{lemma}[Nonvanishing replacement minors]
		\label{lem:supp-coordinate-minor}
		Fix $P$.  Let $K_0\subseteq\Z_N\setminus\{j_P\}$ and
		$J\subseteq\Z_N$ satisfy $\abs{K_0}=\abs J=s\leq N-1$.  Then
		\begin{equation}
			\det\bigl(\gamma_Q^{(P,k)}\bigr)_{k\in K_0,\,Q\in J}\neq0.
			\label{eq:supp-coordinate-minor-det}
		\end{equation}
		In particular, every $\gamma_Q^{(P,k)}$ is nonzero, and for every
		$J$ with $\abs J\leq N-1$ the restricted rows
		$\{\boldsymbol{\gamma}^{(P,k)}_J:k\neq j_P\}$ span
		$(\F_p^J)^*$ (here $V*$ means the dual space of $V$).
	\end{lemma}

	\begin{proof}
		Assume
		$1\leq s\leq N-1$.  Order the two sets as
		\begin{equation}
			K_0=\{k_1,\ldots,k_s\},
			\qquad
			J=\{Q_1,\ldots,Q_s\}.
			\label{eq:supp-coordinate-minor-ordering}
		\end{equation}
		Starting from the baseline transversal matrix
		\begin{equation}
			M=
			\begin{pmatrix}
				\mathbf g_{0,j_0}^{\transpose}\\
				\mathbf g_{1,j_1}^{\transpose}\\
				\vdots\\
				\mathbf g_{N-1,j_{N-1}}^{\transpose}
			\end{pmatrix},
			\label{eq:supp-coordinate-minor-baseline}
		\end{equation}
		construct $M'$ by replacing, for each $a=1,\ldots,s$, the row in
		position $Q_a$ by $\mathbf g_{P,k_a}^{\transpose}$.  Thus
		\begin{equation}
			(M')_{Q,*}
			=
			\begin{cases}
				\mathbf g_{Q,j_Q}^{\transpose},
				& Q\notin J,\\[1mm]
				\mathbf g_{P,k_a}^{\transpose},
				& Q=Q_a\in J.
			\end{cases}
			\label{eq:supp-coordinate-minor-Mprime-rows}
		\end{equation}
		
		The rows of $M'$ correspond to $N$ distinct columns of the MDS
		generator matrix.  Indeed, the inserted columns
		$\mathbf g_{P,k_1},\ldots,\mathbf g_{P,k_s}$ are mutually distinct
		because the $k_a$ are distinct; none of them coincides with the
		baseline column $\mathbf g_{P,j_P}$ because
		$K_0\subseteq\Z_N\setminus\{j_P\}$; and all remaining baseline
		columns have labels $(Q,j_Q)$ with $Q\neq P$.  Hence the MDS
		minor property gives
		\begin{equation}
			\det M'\neq0.
			\label{eq:supp-coordinate-minor-Mprime-nonsingular}
		\end{equation}
		The same property gives $\det M\neq0$.
		
		We now compute $M'M^{-1}$ explicitly.  Since the $Q$th row of
		$M$ is $\mathbf g_{Q,j_Q}^{\transpose}$ and
		$MM^{-1}=I_N$, one has
		\begin{equation}
			\mathbf g_{Q,j_Q}^{\transpose}M^{-1}
			=
			\mathbf e_Q^{\transpose},
			\qquad Q\in\Z_N,
			\label{eq:supp-coordinate-minor-standard-row}
		\end{equation}
		where $\mathbf e_Q$ denotes the $Q$th standard basis vector of
		$\F_p^N$.  For an inserted row, the definition
		\begin{equation}
			\boldsymbol{\gamma}^{(P,k)}
			=
			\mathbf g_{P,k}^{\transpose}M^{-1}
			=
			\bigl(
			\gamma_0^{(P,k)},\ldots,
			\gamma_{N-1}^{(P,k)}
			\bigr)
			\label{eq:supp-coordinate-minor-gamma-row}
		\end{equation}
		gives
		\begin{equation}
			\mathbf g_{P,k_a}^{\transpose}M^{-1}
			=
			\boldsymbol{\gamma}^{(P,k_a)}.
			\label{eq:supp-coordinate-minor-inserted-row}
		\end{equation}
		Consequently,
		\begin{equation}
			(M'M^{-1})_{Q,*}
			=
			\begin{cases}
				\mathbf e_Q^{\transpose},
				& Q\notin J,\\[1mm]
				\boldsymbol{\gamma}^{(P,k_a)},
				& Q=Q_a\in J.
			\end{cases}
			\label{eq:supp-coordinate-minor-product-rows}
		\end{equation}
		Thus $M'M^{-1}$ is obtained from the identity matrix by replacing
		exactly the rows indexed by $J$ with the corresponding
		$\boldsymbol{\gamma}^{(P,k_a)}$.
		
		To display the determinant structure explicitly, write
		\begin{equation}
			J^c
			=
			\{R_1,\ldots,R_{N-s}\}.
			\label{eq:supp-coordinate-minor-complement}
		\end{equation}
		Let $\mathsf P_J$ be the permutation matrix which simultaneously
		reorders the row and column indices from
		$0,1,\ldots,N-1$ to
		\begin{equation}
			R_1,\ldots,R_{N-s},Q_1,\ldots,Q_s.
			\label{eq:supp-coordinate-minor-permutation-order}
		\end{equation}
		Using Eq.~\eqref{eq:supp-coordinate-minor-product-rows}, the
		simultaneously permuted matrix has the block form
		\begin{equation}
			\mathsf P_J(M'M^{-1})\mathsf P_J^{\transpose}
			=
			\begin{pmatrix}
				I_{N-s} & 0\\[1mm]
				\Gamma_{J^c} & \Gamma_J
			\end{pmatrix},
			\label{eq:supp-coordinate-minor-block-form}
		\end{equation}
		where
		\begin{equation}
			\Gamma_{J^c}
			=
			\begin{pmatrix}
				\gamma_{R_1}^{(P,k_1)}
				& \cdots
				& \gamma_{R_{N-s}}^{(P,k_1)}
				\\
				\vdots
				& \ddots
				& \vdots
				\\
				\gamma_{R_1}^{(P,k_s)}
				& \cdots
				& \gamma_{R_{N-s}}^{(P,k_s)}
			\end{pmatrix}
			\label{eq:supp-coordinate-minor-left-block}
		\end{equation}
		and
		\begin{equation}
			\Gamma_J
			=
			\begin{pmatrix}
				\gamma_{Q_1}^{(P,k_1)}
				& \gamma_{Q_2}^{(P,k_1)}
				& \cdots
				& \gamma_{Q_s}^{(P,k_1)}
				\\
				\gamma_{Q_1}^{(P,k_2)}
				& \gamma_{Q_2}^{(P,k_2)}
				& \cdots
				& \gamma_{Q_s}^{(P,k_2)}
				\\
				\vdots
				& \vdots
				& \ddots
				& \vdots
				\\
				\gamma_{Q_1}^{(P,k_s)}
				& \gamma_{Q_2}^{(P,k_s)}
				& \cdots
				& \gamma_{Q_s}^{(P,k_s)}
			\end{pmatrix}
			=
			\bigl(
			\gamma_{Q_b}^{(P,k_a)}
			\bigr)_{a,b=1}^{s}.
			\label{eq:supp-coordinate-minor-right-block}
		\end{equation}
		The upper-right block in
		Eq.~\eqref{eq:supp-coordinate-minor-block-form} vanishes because,
		for every $R_b\in J^c$, the unreplaced row
		$\mathbf e_{R_b}^{\transpose}$ has support only in the column
		$R_b$, which also belongs to $J^c$.
		
		The matrix in Eq.~\eqref{eq:supp-coordinate-minor-block-form} is
		block lower triangular.  Therefore,
		\begin{equation}
			\det\!\left[
			\mathsf P_J(M'M^{-1})\mathsf P_J^{\transpose}
			\right]
			=
			\det(I_{N-s})\det(\Gamma_J)
			=
			\det(\Gamma_J).
			\label{eq:supp-coordinate-minor-block-det}
		\end{equation}
		On the other hand, since $\mathsf P_J$ is a permutation matrix,
		$\det(\mathsf P_J)=\pm1$, and hence
		\begin{align}
			\det\!\left[
			\mathsf P_J(M'M^{-1})\mathsf P_J^{\transpose}
			\right]
			&=
			\det(\mathsf P_J)
			\det(M'M^{-1})
			\det(\mathsf P_J^{\transpose})
			\nonumber\\
			&=
			\det(\mathsf P_J)^2
			\det(M'M^{-1})
			\nonumber\\
			&=
			\det(M'M^{-1}).
			\label{eq:supp-coordinate-minor-permutation-det}
		\end{align}
		Combining
		Eqs.~\eqref{eq:supp-coordinate-minor-block-det} and
		\eqref{eq:supp-coordinate-minor-permutation-det} yields the exact
		identity
		\begin{equation}
			\det(M'M^{-1})
			=
			\det(\Gamma_J)
			=
			\det
			\bigl(
			\gamma_{Q_b}^{(P,k_a)}
			\bigr)_{a,b=1}^{s}.
			\label{eq:supp-coordinate-minor-exact-det}
		\end{equation}
		There is no sign ambiguity in
		Eq.~\eqref{eq:supp-coordinate-minor-exact-det}, because the same
		permutation has been applied to both rows and columns.
		
		Finally,
		\begin{equation}
			\det(M'M^{-1})
			=
			\det(M')\det(M^{-1})
			=
			\frac{\det M'}{\det M}.
			\label{eq:supp-coordinate-minor-ratio}
		\end{equation}
		Both determinants on the right-hand side are nonzero, and therefore
		\begin{equation}
			\det
			\bigl(
			\gamma_{Q_b}^{(P,k_a)}
			\bigr)_{a,b=1}^{s}
			\neq0.
			\label{eq:supp-coordinate-minor-det-ordered}
		\end{equation}
		This proves the asserted nonvanishing minor.
		
		For the first consequence, take $s=1$,
		$K_0=\{k\}$, and $J=\{Q\}$.  Equation
		\eqref{eq:supp-coordinate-minor-det-ordered} then reduces to
		\begin{equation}
			\gamma_Q^{(P,k)}\neq0
			\qquad
			(k\neq j_P,\; Q\in\Z_N).
			\label{eq:supp-coordinate-minor-entry-nonzero}
		\end{equation}
		
		For the spanning statement, fix a nonempty
		$J\subseteq\Z_N.$ 
		There are exactly $N-1$ choices of
		$k\in\Z_N\setminus\{j_P\}$, so we may choose distinct
		$k_1,\ldots,k_s$ from this set.  Applying the determinant statement
		with
		\begin{equation}
			K_0=\{k_1,\ldots,k_s\}
		\end{equation}
		shows that the $s\times s$ matrix
		\begin{equation}
			\bigl(
			\gamma_Q^{(P,k)}
			\bigr)_{k\in K_0,\,Q\in J}
		\end{equation}
		is nonsingular.  Hence the $s$ restricted rows
		\begin{equation}
			\boldsymbol{\gamma}^{(P,k_1)}_J,
			\ldots,
			\boldsymbol{\gamma}^{(P,k_s)}_J
		\end{equation}
		are linearly independent in the $s$-dimensional space
		$(\F_p^J)^*$ and therefore form a basis of that space.  The full
		family
		\begin{equation}
			\{
			\boldsymbol{\gamma}^{(P,k)}_J:
			k\neq j_P
			\}
		\end{equation}
		consequently spans $(\F_p^J)^*$. 
	\end{proof}

	\begin{lemma}[Prime-field translation]
		\label{lem:supp-prime-translation}
		If $\varnothing\neq E\subseteq\F_p$ and $E+\delta=E$ for some
		$\delta\neq0$, then $E=\F_p$.
	\end{lemma}
	
	\begin{proof}
		Every nonzero element of the additive group of the prime field $\F_p$ has
		order $p$.  Thus for any $e\in E$, translation invariance gives
		$e+m\delta\in E$ for all integers $m$, and these $p$ values exhaust
		$\F_p$.
	\end{proof}
	
	\begin{lemma}[Filling one side of a Cartesian slice]
		\label{lem:supp-one-side-full}
		Assume that $M_{\mathbf j}\Lambda$ is Cartesian across $X|Y$ for every
		address word $\mathbf j$.  Fix one baseline and write
		\begin{equation}
			M\Lambda=A\times B,
			\qquad
			A\subseteq\F_p^X,
			\quad
			B\subseteq\F_p^Y.
			\label{eq:supp-baseline-cartesian}
		\end{equation}
		If $\abs B\geq2$, then $A=\F_p^X$.  Symmetrically, if
		$\abs A\geq2$, then $B=\F_p^Y$.
	\end{lemma}
	
	\begin{proof}
		Assume $\abs B\geq2$ and fix $P\in X$.  Choose distinct
		$\mathbf b,\mathbf b'\in B$.  Since
		$\boldsymbol{\gamma}^{(P,k)}_Y$, $k\neq j_P$, span $(\F_p^Y)^*$,
		there is an alternative address $k$ such that
		\begin{equation}
			\boldsymbol{\gamma}^{(P,k)}_Y\cdot(\mathbf b-\mathbf b')\neq0.
			\label{eq:supp-detect-displacement}
		\end{equation}
		Write $\boldsymbol\gamma=\boldsymbol{\gamma}^{(P,k)}$ and let $M'$ be the
		transversal matrix obtained by replacing only the address $j_P$ by $k$.
		The map $M'M^{-1}$ leaves every coordinate except $P$ unchanged.  Writing
		an $X$-coordinate as $(\mathbf u,s)$ with
		$\mathbf u\in\F_p^{X\setminus\{P\}}$, it acts as
		\begin{equation}
			(\mathbf u,s,\mathbf b)
			\longmapsto
			\left(
			\mathbf u,
			\gamma_Ps+\boldsymbol\gamma_{X\setminus\{P\}}\cdot\mathbf u
			+\boldsymbol\gamma_Y\cdot\mathbf b,
			\mathbf b
			\right).
			\label{eq:supp-coordinate-change}
		\end{equation}
		By Lemma~\ref{lem:supp-coordinate-minor}, $\gamma_P\neq0$.
		
		The new slice image is Cartesian, say $M'\Lambda=A'\times B'$.  Since
		Eq.~\eqref{eq:supp-coordinate-change} leaves the $Y$-coordinate unchanged,
		its $Y$-projection is exactly $B$, hence $B'=B$.  Therefore the $X$-section
		of $A'\times B$ equals the same set $A'$ for every $\mathbf b\in B$.
		For fixed $\mathbf u$, define
		\begin{equation}
			A_{\mathbf u}=\{s\in\F_p:(\mathbf u,s)\in A\}.
		\end{equation}
		The $P$-fiber in the transformed $X$-section associated with $\mathbf b$
		is
		\begin{equation}
			\gamma_PA_{\mathbf u}
			+\boldsymbol\gamma_{X\setminus\{P\}}\cdot\mathbf u
			+\boldsymbol\gamma_Y\cdot\mathbf b.
		\end{equation}
		Equality of the transformed $X$-sections for $\mathbf b$ and
		$\mathbf b'$ therefore implies, for every nonempty $A_{\mathbf u}$,
		\begin{equation}
			A_{\mathbf u}+\delta=A_{\mathbf u},
			\qquad
			\delta=
			\gamma_P^{-1}\boldsymbol\gamma_Y\cdot(\mathbf b'-\mathbf b).
			\label{eq:supp-fiber-translation}
		\end{equation}
		Equation~\eqref{eq:supp-detect-displacement} gives $\delta\neq0$, so
		Lemma~\ref{lem:supp-prime-translation} yields
		$A_{\mathbf u}=\F_p$ for every nonempty fiber.  Since the argument applies
		to every $P\in X$, begin with any point of the nonempty set $A$ and change
		its coordinates successively and arbitrarily.  Every vector of $\F_p^X$
		is reached, so $A=\F_p^X$.  The proof with $X$ and $Y$ interchanged is
		identical.
	\end{proof}
	
	\begin{proof}[Proof of Lemma~\ref{lem:rectangle-rigidity}]
		Assume
		$D_\Lambda=\bigcup_{\mathbf t\in\Lambda}T_{\mathbf t}$ is an $X|Y$
		Cartesian rectangle and $\abs\Lambda\geq2$.  By
		Lemma~\ref{lem:supp-slice-image}, every fixed-address image is Cartesian.
		For one baseline write $M\Lambda=A\times B$.  Since $M$ is invertible,
		\begin{equation}
			\abs\Lambda=\abs A\,\abs B\geq2.
		\end{equation}
		Thus at least one of $A,B$ contains at least two points.  If
		$\abs B\geq2$, Lemma~\ref{lem:supp-one-side-full} gives
		$A=\F_p^X$; since $X\neq\varnothing$, now $\abs A\geq2$, and the symmetric
		part of the same lemma gives $B=\F_p^Y$.  The case $\abs A\geq2$ is
		symmetric.  Hence $M\Lambda=\F_p^N$, and invertibility of $M$ gives
		$\Lambda=\F_p^N$.
	\end{proof}
	
	\begin{lemma}[Projected-tile incidence connectivity]
		\label{lem:incidence}
		Let $\varnothing\neq Z\subsetneq\Z_N$.  Form a graph whose vertices are
		computational cells of the block $Z$, joining two cells when they belong to
		a common projected tile $F_Z(\mathbf t)$.  Then the graph is connected.
	\end{lemma}
	
	\begin{proof}
		It is enough to connect two block cells that differ at a single party
		$P\in Z$, since arbitrary cells can then be connected one party at a time.
		Suppose first that their local addresses at $P$ are different:
		$q_P=(j,r)$ and $q_P'=(j',r')$ with $j\neq j'$.  To place both block cells
		in a common $F_Z(\mathbf t)$, impose
		\begin{equation}
			\mathbf t^{\transpose}\mathbf g_{P,j}=r,
			\qquad
			\mathbf t^{\transpose}\mathbf g_{P,j'}=r',
			\label{eq:supp-two-P-equations}
		\end{equation}
		together with the common equation at every party in $Z\setminus\{P\}$.
		There are $\abs Z+1\leq N$ equations, and their coefficient rows are the
		transposes of distinct generator columns.  These rows are independent.  The induced
		linear map from $\F_p^N$ to $\F_p^{\abs Z+1}$ has full row rank and is
		therefore surjective, so the prescribed right-hand sides have a solution
		$\mathbf t$.  The two cells are adjacent.
		
		If the two local cells have the same address $j$ but different symbols
		$r\neq r'$, choose an auxiliary address $k\neq j$ and any $s\in\F_p$.
		Insert the intermediate local cell $(k,s)$, keeping every other block
		coordinate fixed.  The first part gives an edge from the first cell to the
		intermediate one and another edge from the intermediate cell to the second.
		Thus cells differing at one party are connected by a path of length at most
		two, and concatenating such paths proves connectivity.
	\end{proof}

	\section{Complementary subspace, partial-transpose symmetry, and GME witnesses}
	\label{sec:supp-witness}

	\begin{proposition}[Exact complementary projector]
		\label{prop:supp-complement}
		Let
		\begin{equation}
			\ket{e_{\mathbf t}}=N^{-N/2}\ket{\psi_{\mathbf t}},
			\qquad
			\ket s=D^{-1/2}\ket S.
			\label{eq:supp-normalized-hole-stopper}
		\end{equation}
		Then $\{\ket{e_{\mathbf t}}\}_{\mathbf t\in\F_p^N}$ is orthonormal,
		\begin{equation}
			\ket s=p^{-N/2}\sum_{\mathbf t\in\F_p^N}\ket{e_{\mathbf t}},
			\label{eq:supp-normalized-stopper}
		\end{equation}
		and
		\begin{align}
			\cG_{\cC}
			&=
			\left\{
			\sum_{\mathbf t}a_{\mathbf t}\ket{\psi_{\mathbf t}}:
			\sum_{\mathbf t}a_{\mathbf t}=0
			\right\},
			\label{eq:supp-complement-space}\\
			\Pi_{\cG}
			&=
			\sum_{\mathbf t\in\F_p^N}\proj{e_{\mathbf t}}-\proj{s}.
			\label{eq:supp-complement-projector}
		\end{align}
		In particular, $\dim\cG_{\cC}=p^N-1$.
	\end{proposition}
	
	\begin{proof}
		The supports of the deleted modes are distinct tiles, hence orthogonal, and
		$\norm{\psi_{\mathbf t}}^2=N^N$, proving orthonormality after the stated
		normalization.  The tile partition gives
		$\ket S=\sum_{\mathbf t}\ket{\psi_{\mathbf t}}$, while
		$\norm S^2=D=(Np)^N$; this proves Eq.~\eqref{eq:supp-normalized-stopper}.
		
		If
		$Q=\sum_{\mathbf t}\proj{e_{\mathbf t}}$ is the projector onto the deleted-mode
		span, then $\ket s\in\ran Q$ and $\cG_{\cC}=\ran Q\cap\ket s^\perp$.
		Hence its orthogonal projector is $Q-\proj{s}$, proving
		Eq.~\eqref{eq:supp-complement-projector} and the dimension statement.
	\end{proof}
	
	\begin{proposition}[Exact partial transpose symmetry]
		\label{prop:supp-pt}
		For every subset $X\subseteq\Z_N$,
		\begin{equation}
			\Pi_{\cG}^{T_X}=\Pi_{\cG},
			\qquad
			P_{\cU}^{T_X}=P_{\cU},
			\qquad
			\rho_{\cG}^{T_X}=\rho_{\cG}.
			\label{eq:supp-pt-invariance}
		\end{equation}
	\end{proposition}
	
	\begin{proof}
		For each $\mathbf t$,
		\begin{equation}
			\ket{e_{\mathbf t}}
			=
			\bigotimes_{P\in\Z_N}
			\left(
			N^{-1/2}\sum_{j\in\Z_N}
			\ket{j,c_{P,j}(\mathbf t)}_P
			\right),
			\label{eq:supp-hole-real-product}
		\end{equation}
		so every local factor has real coefficients in the computational basis.
		The normalized stopper $\ket s$ is likewise a fully product vector with
		real local coefficients.  Therefore, for every subset $X$,
		\begin{equation}
			\bigl(\proj{e_{\mathbf t}}\bigr)^{T_X}=\proj{e_{\mathbf t}},
			\qquad
			(\proj{s})^{T_X}=\proj{s}.
		\end{equation}
		Applying $T_X$ to Eq.~\eqref{eq:supp-complement-projector} gives
		$\Pi_{\cG}^{T_X}=\Pi_{\cG}$.  Since $P_{\cU}=I-\Pi_{\cG}$ and
		$\rho_{\cG}=\Pi_{\cG}/R$, the remaining identities follow immediately.
	\end{proof}
	
	\begin{proposition}[Positive biseparable threshold]
		\label{prop:supp-positive-gap}
		The quantity $\epsilon_G$ in Eq.~\eqref{eq:biseparable-threshold} is
		strictly positive.
	\end{proposition}
	
	\begin{proof}
		For each nontrivial bipartition $X|Y$, let $\mathcal P_{X|Y}$ be the compact
		set of normalized $X|Y$ product vectors.  Their finite union
		$\mathcal P_{\rm bi}=\bigcup_{X|Y}\mathcal P_{X|Y}$ is compact.  The
		continuous function $f(\phi)=\bra\phi P_{\cU}\ket\phi$ has no zero on
		$\mathcal P_{\rm bi}$ by Theorem~\ref{thm:gupb-existence}; hence it has a
		strictly positive minimum there.  The biseparable state set is the convex
		hull of the corresponding rank-one projectors, and
		$\Tr(P_{\cU}\,\cdot)$ is linear.  Its minimum over the convex hull therefore
		equals the minimum over the generating product states, proving
		$\epsilon_G>0$.
	\end{proof}

	\begin{proposition}[Completion of Theorem~\ref{thm:ppt-gme-witness}]
		\label{prop:supp-ppt-gme-family}
		All statements of Theorem~\ref{thm:ppt-gme-witness} hold.
	\end{proposition}
	
	\begin{proof}
		Proposition~\ref{prop:supp-positive-gap} shows that
		$W_G=P_{\cU}-\epsilon_G I$ is nonnegative on every biseparable state.
		Because $P_{\cU}\Pi_{\cG}=0$ and $\Tr\rho_{\cG}=1$,
		\begin{equation}
			\Tr(W_G\rho_{\cG})=-\epsilon_G<0.
			\label{eq:supp-witness-on-rhog}
		\end{equation}
		Thus $W_G$ is a GME witness detecting $\rho_{\cG}$.
		
		Fix a nontrivial bipartition $X|Y$.  If $W_G$ were decomposable with
		respect to this bipartition, $W_G=A_X+B_X^{T_X}$ with
		$A_X,B_X\geq0$.  Proposition~\ref{prop:supp-pt} would then give
		\begin{equation}
			\Tr(W_G\rho_{\cG})
			=
			\Tr(A_X\rho_{\cG})+
			\Tr(B_X\rho_{\cG}^{T_X})
			\geq0,
		\end{equation}
		contradicting Eq.~\eqref{eq:supp-witness-on-rhog}.  Since the bipartition
		was arbitrary, $W_G$ is nondecomposable with respect to every bipartition.
		
		Now consider
		\begin{equation}
			\rho_\lambda
			=
			\frac{\Pi_{\cG}+\lambda P_{\cU}}{R+\lambda K},
			\qquad 0<\lambda\leq1.
		\end{equation}
		Because $\Pi_{\cG}$ and $P_{\cU}$ are complementary orthogonal projectors,
		$\rho_\lambda$ has strictly positive eigenvalues
		$(R+\lambda K)^{-1}$ on $\cG_{\cC}$ and
		$\lambda(R+\lambda K)^{-1}$ on $\Span(\cU_{\cC})$.  Hence it is full rank.
		Proposition~\ref{prop:supp-pt} gives
		$\rho_\lambda^{T_X}=\rho_\lambda>0$ for every bipartition.  Moreover,
		\begin{equation}
			\Tr(W_G\rho_\lambda)
			=
			\frac{\lambda K}{R+\lambda K}-\epsilon_G.
			\label{eq:supp-rholambda-witness}
		\end{equation}
		Since  
		$\epsilon_G <1$, this expectation is negative exactly when
		\begin{equation}
			0<\lambda<\frac{R\epsilon_G}{K(1-\epsilon_G)}=\kappa_G.
		\end{equation}
		Thus every state in the stated interval is GME.
		
		Finally, if $W_{\rm fd}$ is fully decomposable, then for every bipartition
		it admits $W_{\rm fd}=A_X+B_X^{T_X}$ with $A_X,B_X\geq0$.  Since every
		$\rho_\lambda$ is PPT with respect to every bipartition,
		\begin{equation}
			\Tr(W_{\rm fd}\rho_\lambda)
			=
			\Tr(A_X\rho_\lambda)+
			\Tr(B_X\rho_\lambda^{T_X})
			\geq0.
		\end{equation}
		Hence no fully decomposable witness detects any member of this family,
		whereas $W_G$ detects precisely the interval stated in the theorem.
	\end{proof}

	\section{Measurement, one-copy discrimination}
	\label{sec:supp-measurement}

	The first measurement statement is a cone separation for the complementary
	projector itself.
	
	\begin{proposition}[PPT across every bipartition but outside the biseparable cone]
		\label{prop:supp-ppt-bsep-separation}
		The complementary projector satisfies
		\begin{equation}
			\Pi_{\cG}
			\in
			\left(\bigcap_{X|Y}\PPT_{X|Y}^{+}\right)
			\setminus\BSEP^{+},
			\label{eq:supp-ppt-bsep-separation}
		\end{equation}
		where $\PPT_{X|Y}^{+}$ is the cone of positive operators with positive
		partial transpose across $X|Y$.
	\end{proposition}
	
	\begin{proof}
		Proposition~\ref{prop:supp-pt} gives
		$\Pi_{\cG}^{T_X}=\Pi_{\cG}\geq0$ for every bipartition, proving membership in all
		the PPT cones. Suppose instead that
		$\Pi_{\cG}=\sum_{\mu=1}^m X_\mu$, where every nonzero $X_\mu\geq0$ is
		separable across some bipartition that may depend on $\mu$. For a fixed $\mu$,
		$0\leq X_\mu\leq\Pi_{\cG}$. If $\ket z\in\ker\Pi_{\cG}$, then
		\begin{equation}
			0
			\leq
			\bra zX_\mu\ket z
			\leq
			\bra z\Pi_{\cG}\ket z
			=0.
			\label{eq:supp-kernel-order}
		\end{equation}
		Positivity gives $X_\mu^{1/2}\ket z=0$, so
		$\ker\Pi_{\cG}\subseteq\ker X_\mu$ and therefore
		$\ran X_\mu\subseteq\ran\Pi_{\cG}=\cG_{\cC}$.
		
		Choose a nonzero summand $X_\mu$ and write one of its finite separable
		decompositions as $X_\mu=\sum_s\omega_s\proj{a_s,b_s}$ with
		$\omega_s>0$. Its range is the span of the product vectors
		$\ket{a_s,b_s}$, so at least one nonzero product vector for that bipartition belongs to
		$\cG_{\cC}$. This contradicts Theorem~\ref{thm:gupb-existence}. Hence no such
		biseparable decomposition of $\Pi_{\cG}$ exists.
	\end{proof}
	
	\subsection{One-copy separable discrimination}
	
	We next prove the following criterion and then the quantitative sharp
	tradeoff used in Eq.~\eqref{eq:one-copy-biseparable-bound}.
	
	\begin{proposition}[One-sided unambiguous discrimination criterion]
		\label{prop:supp-one-sided-discrimination}
		Fix a nontrivial bipartition $X|Y$. There exists a nonzero $X|Y$-separable effect
		$E$ such that
		\begin{equation}
			\Tr(E\rho_{\cU})=0,
			\qquad
			\Tr(E\rho_{\cG})>0
			\label{eq:supp-one-sided-error-free}
		\end{equation}
		if and only if $\cG_{\cC}$ contains an $X|Y$ product vector.
	\end{proposition}
	
	\begin{proof}
		If $\ket{a,b}\in\cG_{\cC}$ is nonzero, then
		set $\ket{\widehat a}=\ket a/\norm a$ and
		$\ket{\widehat b}=\ket b/\norm b$. Then
		$E=\proj{\widehat a,\widehat b}$ is a rank-one product effect with zero
		overlap with $P_{\cU}$ and positive overlap with $\Pi_{\cG}$. Conversely, let
		$E\geq0$ be a nonzero separable effect with $\Tr(EP_{\cU})=0$. Since
		$P_{\cU}$ is an orthogonal projector,
		\begin{equation}
			0
			=
			\Tr(E^{1/2}P_{\cU}E^{1/2})
			=
			\norm{P_{\cU}E^{1/2}}_{\mathrm F}^{2},
			\label{eq:supp-zero-trace-range}
		\end{equation}
		so $P_{\cU}E^{1/2}=0$. The positive operators $E$ and $E^{1/2}$ have the
		same range, hence $\ran E\subseteq\ker P_{\cU}=\cG_{\cC}$. By separability,
		$E=\sum_s\omega_s\proj{a_s,b_s}$ for nonzero product vectors and positive
		coefficients. The range of this positive sum is
		$\Span\{\ket{a_s,b_s}\}_s$; consequently every summand vector, and in
		particular at least one nonzero product vector, belongs to
		$\ran E\subseteq\cG_{\cC}$. The condition
		$\Tr(E\rho_{\cG})>0$ follows automatically from $E\neq0$ and
		$\ran E\subseteq\ran\Pi_{\cG}$.
	\end{proof}
	
	\begin{proof}[Proof of the one-copy tradeoff Eq.\eqref{eq:one-copy-biseparable-bound}]
		Any biseparable positive operator can be written as a finite conic
		combination
		\begin{equation}
			E
			=
			\sum_s\omega_s\proj{\phi_s},
			\qquad
			\omega_s>0,
			\label{eq:supp-bsep-effect-decomposition}
		\end{equation}
		where each normalized $\ket{\phi_s}$ is product with respect to some bipartition
		that may depend on $s$. Put
		\begin{equation}
			p_s
			=
			\bra{\phi_s}P_{\cU}\ket{\phi_s}.
			\label{eq:supp-ps}
		\end{equation}
		Proposition~\ref{prop:supp-positive-gap} gives
		$p_s\geq\epsilon_G$, while
		\begin{equation}
			\bra{\phi_s}\Pi_{\cG}\ket{\phi_s}
			=
			1-p_s.
			\label{eq:supp-ds}
		\end{equation}
		Because $\frac{x}{1-x}$ is increasing on $[0,1)$,
		\begin{equation}
			\frac{p_s}{1-p_s}
			\geq
			\frac{\epsilon_G}{1-\epsilon_G}
			\label{eq:supp-ratio-gap}
		\end{equation}
		whenever $p_s<1$; if $p_s=1$ the desired inequality is trivial (i.e., the equality \eqref{eq:supp-termwise-onecopy} ).
		Multiplying Eq.~\eqref{eq:supp-ratio-gap} by
		$\omega_s(1-p_s)/K$ gives
		\begin{equation}
			\omega_s\frac{p_s}{K}
			\geq
			\frac{R\epsilon_G}{K(1-\epsilon_G)}
			\omega_s\frac{1-p_s}{R}.
			\label{eq:supp-termwise-onecopy}
		\end{equation}
		Summing over $s$ yields Eq.~\eqref{eq:one-copy-biseparable-bound}.
		
		To prove optimality, compactness in
		Proposition~\ref{prop:supp-positive-gap} supplies a normalized vector that is product with respect to some bipartition $\ket{\phi_*}$ with
		$\bra{\phi_*}P_{\cU}\ket{\phi_*}=\epsilon_G$. The upper bound
		$\epsilon_G <1$  implies that its overlap with the complementary subspace is nonzero. The effect $E_*=\proj{\phi_*}$ then satisfies
		\begin{equation}
			\frac{\Tr(E_*\rho_{\cU})}{\Tr(E_*\rho_{\cG})}
			=
			\frac{R\epsilon_G}{K(1-\epsilon_G)}
			=
			\kappa_G,
			\label{eq:supp-one-copy-attainment}
		\end{equation}
		so the constant cannot be increased.
	\end{proof}

	\section{Proof of Theorem~\ref{thm:tensor-separation}}
	\label{sec:supp-tensor}

	The tensor statement in the Letter uses only the standard bipartite closure
	of UPBs (see also Ref. \cite{DiVincenzo2003}).  We include a proof to make the multipartite application
	self-contained.
	
	\begin{lemma}[Tensor product of bipartite UPBs]
		\label{lem:supp-tensor-upb}
		Let
		$\{\ket{a_i}\otimes\ket{b_i}\}_{i=1}^{K_1}$ be a bipartite UPB in
		$\cH_A\otimes\cH_B$, and let
		$\{\ket{c_j}\otimes\ket{d_j}\}_{j=1}^{K_2}$ be a bipartite UPB in
		$\cH_C\otimes\cH_D$. Then
		\begin{equation}
			\left\{
			(\ket{a_i}\otimes\ket{c_j})_{AC}
			\otimes
			(\ket{b_i}\otimes\ket{d_j})_{BD}
			\right\}_{i,j}
			\label{eq:supp-tensor-upb-set}
		\end{equation}
		is a UPB with respect to $AC|BD$.
	\end{lemma}
	
	\begin{proof}
		For two distinct pairs $(i,j)\neq(i',j')$, either $i\neq i'$ or
		$j\neq j'$. In the first case the inner product contains the zero factor
		$\braket{a_i}{a_{i'}}\braket{b_i}{b_{i'}}$; in the second it contains the zero
		factor $\braket{c_j}{c_{j'}}\braket{d_j}{d_{j'}}$. Thus the displayed family is
		orthogonal. Moreover,
		\begin{equation}
			K_1K_2
			<
			(\dim\cH_A\dim\cH_B)(\dim\cH_C\dim\cH_D),
			\label{eq:supp-tensor-incomplete-count}
		\end{equation}
		because both input UPBs are incomplete, so the product family is
		incomplete. Suppose, for contradiction, that a nonzero product vector
		$\ket x_{AC}\otimes\ket y_{BD}$ is orthogonal to every vector in
		Eq.~\eqref{eq:supp-tensor-upb-set}. For every $i$, define contractions
		\begin{equation}
			\ket{x_i}_C
			=
			(\bra{a_i}\otimes I_C)\ket x,
			\qquad
			\ket{y_i}_D
			=
			(\bra{b_i}\otimes I_D)\ket y.
			\label{eq:supp-contractions}
		\end{equation}
		For every $j$,
		\begin{equation}
			\braket{c_j}{x_i}
			\braket{d_j}{y_i}
			=
			0.
			\label{eq:supp-contracted-orthogonality}
		\end{equation}
		If both $\ket{x_i}$ and $\ket{y_i}$ were nonzero, then
		$\ket{x_i}\otimes\ket{y_i}$ would be a nonzero product vector orthogonal to
		the entire second UPB, impossible. Hence for every $i$,
		\begin{equation}
			\ket{x_i}=0
			\qquad\text{or}\qquad
			\ket{y_i}=0.
			\label{eq:supp-contraction-zero}
		\end{equation}
		
		Let
		$X_A=\ran[\Tr_C(\proj{x})]$ be the support of the reduced positive
		operator of $\ket x$ on $A$, and let
		$Y_B=\ran[\Tr_D(\proj{y})]$ be the analogous support on $B$; here
		$\Tr_C$ and $\Tr_D$ denote partial traces. To verify the first equivalence
		below, write a Schmidt decomposition
		$\ket x=\sum_r s_r\ket{u_r}_A\ket{v_r}_C$ with every $s_r>0$. Then
		$\ket{x_i}=\sum_r s_r\braket{a_i}{u_r}\ket{v_r}$ vanishes exactly when
		$\ket{a_i}$ is orthogonal to every $\ket{u_r}$, that is, to $X_A$.
		The same reasoning on $B|D$ gives
		\begin{equation}
			\ket{x_i}=0
			\Longleftrightarrow
			\ket{a_i}\perp X_A,
			\qquad
			\ket{y_i}=0
			\Longleftrightarrow
			\ket{b_i}\perp Y_B.
			\label{eq:supp-schmidt-support}
		\end{equation}
		Because $\ket x$ and $\ket y$ are nonzero, both $X_A$ and $Y_B$ are
		nonzero. Choose nonzero
		$\ket\alpha\in X_A$ and $\ket\beta\in Y_B$. For every $i$,
		Eq.~\eqref{eq:supp-contraction-zero} together with
		Eq.~\eqref{eq:supp-schmidt-support} gives
		\begin{equation}
			\braket{a_i}{\alpha}
			\braket{b_i}{\beta}
			=
			0.
			\label{eq:supp-first-upb-extension}
		\end{equation}
		Thus $\ket\alpha\otimes\ket\beta$ is a nonzero product vector orthogonal to
		the first UPB, contradiction.
	\end{proof}
	
	\begin{corollary}[Tensor-power stability]
		\label{cor:supp-partywise-stability}
		If $\cU$ and $\cV$ are GUPBs on the same set of parties, their tensor product,
		with tensor factors grouped according to the original parties,
		\begin{equation}
			\cU\boxtimes\cV
			=
			\left\{
			\bigotimes_{P\in\Z_N}
			(\ket{u_P}\otimes\ket{v_P}):
			\bigotimes_P\ket{u_P}\in\cU,
			\ \bigotimes_P\ket{v_P}\in\cV
			\right\}
			\label{eq:supp-partywise-product-definition}
		\end{equation}
		is a GUPB. In particular,
		$\cU_{\cC}^{\otimes\ell}$ is a GUPB for every $\ell\geq1$.
	\end{corollary}
	
	\begin{proof}
		Fix any bipartition $X|Y$. After grouping all systems in $X$ and all systems in
		$Y$, both $\cU$ and $\cV$ are bipartite UPBs with respect to that bipartition. Apply
		Lemma~\ref{lem:supp-tensor-upb}, with the tensor factors regrouped according to the original parties. The product family is therefore a bipartite UPB with respect to this arbitrary bipartition. Hence it is a GUPB.
	\end{proof}

	We now prove positivity, optimality, and the upper bound for
	$\kappa_{\ell,X|Y}$.
	
	\begin{proposition}[Finite-copy discrimination bound]
		\label{prop:supp-finite-copy-slope}
		For every $\ell$ and bipartition $X|Y$, the quantity
		$\kappa_{\ell,X|Y}$ in Eq.~\eqref{eq:tensor-slope-definition} satisfies
		Eq.~\eqref{eq:kappa-range}; every $X|Y$-separable effect satisfies
		Eq.~\eqref{eq:tensor-slope}; and equality is attainable by a scaled
		rank-one product effect.
	\end{proposition}
	
	\begin{proof}
		By Corollary~\ref{cor:supp-partywise-stability},
		$\cU_{\cC}^{\otimes\ell}$ is a bipartite UPB with respect to $X|Y$. Hence
		\begin{equation}
			\epsilon_{\ell,X|Y}
			=
			\min_{\norm A=\norm B=1}
			p_{\cU}^{(\ell)}(A,B)
			>
			0.
			\label{eq:supp-copy-product-gap}
		\end{equation}
		The set of normalized $X|Y$ product vectors is compact, and both $p_{\cU}^{(\ell)}$ and
		$p_{\cG}^{(\ell)}$ are continuous on it. The admissible set $p_{\cG}^{(\ell)}>0$ is nonempty:
		otherwise the Haar average of $p_{\cG}^{(\ell)}$ computed below would vanish. Choose
		one admissible point to see that the infimum of $p_{\cU}^{(\ell)}/p_{\cG}^{(\ell)}$ is finite,
		and let $C<\infty$ bound the ratios along a minimizing sequence. Since
		$p_{\cU}^{(\ell)}\geq\epsilon_{\ell,X|Y}$, every member of that sequence satisfies
		\begin{equation}
			p_{\cG}^{(\ell)}
			\geq
			\frac{\epsilon_{\ell,X|Y}}{C}.
			\label{eq:supp-ratio-away-from-boundary}
		\end{equation}
		It therefore lies in a compact subset on which $p_{\cG}^{(\ell)}$ is bounded away
		from zero. A convergent subsequence attains the minimum. Moreover,
		$0<p_{\cG}^{(\ell)}\leq1$ and $p_{\cU}^{(\ell)}\geq\epsilon_{\ell,X|Y}>0$ at every admissible
		point, so the attained minimum is strictly positive.
		
		To prove the upper bound, average over independent Haar-random normalized
		$\ket A$ and $\ket B$. Since the average rank-one projector is the
		maximally mixed state on $\cH_X^{\otimes\ell}\otimes\cH_Y^{\otimes\ell}$,
		\begin{equation}
			\mathbb E[p_{\cU}^{(\ell)}]
			=
			\frac{K^\ell}{D^\ell},
			\qquad
			\mathbb E[p_{\cG}^{(\ell)}]
			=
			\frac{R^\ell}{D^\ell}.
			\label{eq:supp-haar-averages}
		\end{equation}
		Set $\theta_\ell=(K/R)^\ell$. If the minimum ratio were larger than
		$\theta_\ell$, the continuous function $p_{\cU}^{(\ell)}-\theta_\ell p_{\cG}^{(\ell)}$ would be nonnegative everywhere and
		strictly positive wherever $p_{\cG}^{(\ell)}>0$. Equation
		\eqref{eq:supp-haar-averages} shows that $p_{\cG}^{(\ell)}>0$ at some point. By
		continuity, $p_{\cU}^{(\ell)}-\theta_\ell p_{\cG}^{(\ell)}$ is then positive on a nonempty open set, which
		has positive Haar measure. Its Haar average would be strictly positive. On
		the other hand, Eq.~\eqref{eq:supp-haar-averages} gives
		\begin{equation}
			\mathbb E\,[p_{\cU}^{(\ell)}-\theta_\ell p_{\cG}^{(\ell)}]
			=
			\frac{K^\ell}{D^\ell}
			-
			\left(\frac KR\right)^\ell
			\frac{R^\ell}{D^\ell}
			=0,
			\label{eq:supp-haar-zero-difference}
		\end{equation}
		a contradiction. Thus the minimum ratio is at most $(K/R)^\ell$, proving
		$\kappa_{\ell,X|Y}\leq1$.
		
		Let an $X|Y$-separable effect be decomposed as
		\begin{equation}
			E
			=
			\sum_s\omega_s\proj{A_s,B_s},
			\qquad
			\omega_s>0,
			\label{eq:supp-separable-effect}
		\end{equation}
		with normalized product vectors. The definition of
		$\kappa_{\ell,X|Y}$ gives, term by term,
		\begin{equation}
			\frac{p_{\cU}^{(\ell)}(A_s,B_s)}{K^\ell}
			\geq
			\kappa_{\ell,X|Y}
			\frac{p_{\cG}^{(\ell)}(A_s,B_s)}{R^\ell}.
			\label{eq:supp-termwise-slope}
		\end{equation}
		If $p_{\cG}^{(\ell)}(A_s,B_s)=0$, the right-hand side is zero and the inequality is
		automatic. Summing Eq.~\eqref{eq:supp-termwise-slope} over $s$ gives
		Eq.~\eqref{eq:tensor-slope}.
		
		Finally, let $\ket{A_*,B_*}$ be a normalized product vector attaining the
		minimum ratio; its overlap with the complementary subspace is positive by the definition of the minimization domain. For
		any $0<\lambda\leq1$, the operator
		$E_*=\lambda\proj{A_*,B_*}$ is a valid effect and saturates the termwise
		inequality. Hence the constant is sharp.
	\end{proof}

	It remains to derive the one-copy expression, submultiplicativity, and the
	existence of the exponent stated in the Letter.
	
	For $\ell=1$,
	\begin{equation}
		p_{\cG}^{(1)}(A,B)=1-p_{\cU}^{(1)}(A,B),
		\label{eq:supp-onecopy-complement}
	\end{equation}
	and 
	$\epsilon_{X|Y} <1$, so a minimizer of $p_{\cU}^{(1)}$
	has positive complementary overlap.  Since $p/(1-p)$ is increasing on
	$[0,1)$, minimizing the ratio gives Eq.~\eqref{eq:one-copy-slope}. If minimizers for $\ell$ and $r$ copies are combined according to the original party grouping, then
	\begin{equation}
		p_{\cU}^{(\ell+r)}=p_{\cU}^{(\ell)}p_{\cU}^{(r)},
		\qquad
		p_{\cG}^{(\ell+r)}=p_{\cG}^{(\ell)}p_{\cG}^{(r)}.
		\label{eq:supp-probability-product}
	\end{equation}
	Substitution into Eq.~\eqref{eq:tensor-slope-definition} gives
	Eq.~\eqref{eq:slope-submultiplicative}. Thus
	$a_\ell=-\log\kappa_{\ell,X|Y}$ is superadditive:
	\begin{equation}
		a_{\ell+r}\geq a_\ell+a_r.
		\label{eq:supp-superadditive}
	\end{equation}
	The superadditive form of Fekete's lemma yields
	\begin{equation}
		\lim_{\ell\to\infty}\frac{a_\ell}{\ell}
		=
		\sup_{\ell\geq1}\frac{a_\ell}{\ell},
		\label{eq:supp-fekete}
	\end{equation}
	possibly $+\infty$, which proves Eq.~\eqref{eq:slope-exponent}.
	
	Finally, $\rho_{\cU}^{\otimes\ell}$ is supported on
	$P_{\cU}^{\otimes\ell}$ while
	$\rho_{\cG}^{\otimes\ell}$ is supported on
	$\Pi_{\cG}^{\otimes\ell}$, which is orthogonal to
	$P_{\cU}^{\otimes\ell}$. Hence the measurement in
	Eq.~\eqref{eq:tensor-ppt-measurement} distinguishes them perfectly.
	Equation~\eqref{eq:pt-invariance} makes both effects invariant under partial transposition across every bipartition. By Eq.~\eqref{eq:tensor-slope}, no
	measurement separable across a bipartition can be perfect.

	\begin{proof}[Completion of the proof of Theorem~\ref{thm:tensor-separation}]
		Corollary~\ref{cor:supp-partywise-stability} proves that
		$\cU_{\cC}^{\otimes\ell}$ is a GUPB for every finite $\ell$.
		Proposition~\ref{prop:supp-finite-copy-slope} proves
		$0<\kappa_{\ell,X|Y}\leq1$, the separable-effect inequality
		\eqref{eq:tensor-slope}, and attainment of the constant.
		Equation~\eqref{eq:supp-onecopy-complement} gives the one-copy expression
		\eqref{eq:one-copy-slope}, while
		Eq.~\eqref{eq:supp-probability-product} gives
		\eqref{eq:slope-submultiplicative} and hence the existence of the exponent
		\eqref{eq:slope-exponent} by Fekete's lemma.
		
		It remains only to verify the perfect measurement.  The state
		$\rho_{\cU}^{\otimes\ell}$ is supported on
		$\ran P_{\cU}^{\otimes\ell}$, whereas
		$\rho_{\cG}^{\otimes\ell}$ is supported on
		$\ran\Pi_{\cG}^{\otimes\ell}$, which is orthogonal to
		$\ran P_{\cU}^{\otimes\ell}$.  Therefore
		$\{P_{\cU}^{\otimes\ell},I-P_{\cU}^{\otimes\ell}\}$ distinguishes the two
		states with certainty.  Equation~\eqref{eq:pt-invariance} implies
		$(P_{\cU}^{\otimes\ell})^{T_X}=P_{\cU}^{\otimes\ell}$ for every
		bipartition, and hence the complementary effect is invariant as well.
		This completes every assertion of Theorem~\ref{thm:tensor-separation}.
	\end{proof}
	
\end{document}